\documentclass[confernce]{IEEEtran}
\usepackage{amsfonts,bbm}
\usepackage{epsfig,amsmath,graphics,mathabx,graphicx,mathrsfs,amsthm}
\usepackage{eepic,bm,epsfig,color,subfigure}
\usepackage{hyperref}
\DeclareMathAlphabet{\mathpzc}{OT1}{pzc}{m}{it}

\def\cstar{C^{*}}
\def\comp{\mathbb{C}}

\newcommand{\beq}{\begin{equation}}
\newcommand{\eeq}{\end{equation}}
\newcommand {\cali}[1]{{\mathcal #1}}

\newcommand{\tr}{{\tt Tr}} 
 
\newcommand{\conj}[1]{\overline{#1}}
 
\newcommand{\pmat}{\begin{pmatrix}} 
\newcommand{\emat}{\end{pmatrix}} 

\newcommand{\real}{\mathbb{R}}
\newcommand{\norm}[1]{|\!|#1|\!|} 
\newcommand{\tensor}{\otimes}
\newcommand{\unit}{\mathbbm{1}}

\newcommand{\Sp}{{\tt sp}}
\newtheorem{propn}{Proposition}{}{}
\newtheorem{thm}{Theorem}{}{}
\newtheorem{lem}{Lemma}{}{}
{}{}
{}{}
{}{}
{}
\newcommand{\commentout}[1]{}
\newcommand{\be}{\begin{enumerate}}
\newcommand{\ee}{\end{enumerate}}
\newcommand{\bi}{\begin{itemize}}
\newcommand{\ei}{\end{itemize}}

\begin{document}
\title{An algebraic approach to information theory} 
\author{\IEEEauthorblockN{Manas K. Patra and Samuel L. Braunstein\\}  
\IEEEauthorblockA{Department of Computer Science, University of York, York YO10 5DD, UK}
}
\date{}
\maketitle
\begin{abstract}
This work proposes an algebraic model for classical information theory. We first give an algebraic model of probability theory. Information theoretic constructs are based on this model. In addition to theoretical insights provided by our model one obtains new computational and analytical tools. Several important theorems of classical probability and information theory are presented in the algebraic framework.
\end{abstract}
\section{Introduction}
The present paper reports a brief synopsis of our work on an algebraic model of {\em classical} information theory based on operator algebras. Let us recall a simple model of a communication system proposed by Shanon \cite{Shannon48}. This model has essentially four components: source, channel, encoder/decoder and receiver. Some amount of noise affects every stage of the operation and the behavior of components are generally modeled as stochastic processes. In this work our primary focus will be on discrete processes. A discrete source can be viewed as a generator of a countable set of random variables. In a communication process the source generates sequence of random variables. Then it is sent through the channel (with encoding/decoding) and the output at the receiver is another sequence of random variables. Thus, the concrete objects or  {\em observables}, to use the language of quantum theory, are modeled as random variables. The underlying probability space is primarily used to define probability distributions or {\em states} associated with the relevant random variables. In the algebraic approach we directly model the observables. Since random variables can be added and multiplied \footnote{We assume that they are real or complex valued.}they constitute an {\em algebra}. This is our starting point. In fact, the algebra of random variables have a richer structure called a $\cstar$ algebra. Starting with a $\cstar$ algebra of observables we can define most important concepts in probability theory in general and information theory in particular. A natural question is: why should we adopt this algebraic approach? We discuss the reasons below. 

First, it seems more appropriate to deal with the ``concrete'' quantities, {\em viz}.\ observables and their intrinsic structure. The choice of underlying probability space is somewhat arbitrary as a comparison of standard textbooks on information theory \cite{CoverT,Ciszar} reveals. Moreover, from the algebra of observables we can recover particular probability spaces from representations of the algebra. Second, some constraints, may have to be imposed on the set of random variables. In security protocols  different participants have access to different sets of observables and may assign different probability structures. In this case, the algebraic approach seems more natural: we have to study different subalgebras. Third, the algebraic approach gives us new theoretical insights and computational tools. This will be justified in the following sections. Finally, and this was our original motivation, the algebraic approach provides the basic framework for a unified approach to classical {\em and} quantum information. All quantum protocols have some classical components, e.g.\ classical communication, ``coin-tosses'' etc. But the language of the two processes, classical and quantum, seem quite different. In the former we are dealing with random variables defined on one or more probability spaces where as in the latter we are processing quantum states which also give complete information about the measurement statistics of {\em quantum} observables. The algebraic framework is eminently suitable for bringing together these somewhat disparate viewpoints. Classical observables are simply elements that commute with every element in the algebra. 

The connection between operator algebras and information theory---classical {\em and} quantum---have appeared in the scientific literature since the beginnings of information theory and operator algebras---both classical and quantum (see e.g.\ \cite{Umegaki4, Segal60,araki75,Keyl,Beny,Kretschmann}). Most previous work focus on some aspects of information theory like the noncommutative generalizations of the concepts of entropy. There does not appear to be a unified and coherent approach based on intrinsically algebraic notions. The construction of such a model is one of the goals of the paper. As probabilistic concepts play such an important role in the development of information theory we first present an algebraic approach to probability. I. E. Segal \cite{Segal54} first proposed such an algebraic approach model of  probability theory. Later Voiculescu \cite{Voic} developed noncommutative or ``free probability'' theory. We believe several aspects of our approach are novel and yield deeper insights to information processes. In this summary, we have omitted most proofs or give only brief outlines. The full proofs can be found in our \href{http://arxiv.org/abs/0910.1536}{arXiv submission} \cite{Patra09}. A brief outline of the paper follows. 

In Section \ref{sec:algebra} we give the basic definitions of the $\cstar$ algebras.  This is followed by an account of probabilistic concepts from an algebraic perspective. In particular, we investigate the fundamental notion of independence and demonstrate how it relates to the algebraic structure. One important aspect in which our approach seems novel is the treatment of  probability distribution functions.  In Section \ref{sec:info} we give a precise algebraic model of information/communication system. The fundamental concept of entropy is introduced. We also define and study the crucial notion of a channel as a (completely) positive map. In particular, the {\em channel coding theorem} is presented as an approximation result. Stated informally: {\em Every channel other than the useless ones can be approximated by a lossless channel under appropriate coding}. We conclude the paper with some comments and discussions.  

\section{$C^*$ Algebras and Probability} \label{sec:algebra}
A Banach algebra $A$ is a complete normed algebra \cite{Rudin,KR1}. That is, $A$ is an algebra over real ($\real$) or complex numbers ($\comp$), for every $x\in A$ the norm $\norm{x}\geq 0$ is defined satisfying the usual properties and every Cauchy sequence converges in the norm. 
A $\cstar$ algebra $B$ is a Banach algebra\cite{KR1} with 
an anti-linear involution $^*$ ($x^{**}=x$ and $(x+cy)^*=x^*+\conj{c}y^*$, $x,y\in B$ and $c\in\comp$) such that 
\(\norm{xx^*}=\norm{x}^2\text{ and } (xy)^*=y^*x^*\forall x,y \in B\).
This implies that $\norm{x}=\norm{x^*}$. We often assume that the unit $I\in B$.  
The fundamental Gelfand-Naimark-Segal ({\bf GNS}) theorem states that 
every $\cstar$ algebra can be isometrically embedded in some $\cali{L}(H)$, the set of bounded operators on a Hilbert space of $H$. The spectrum of an element $x\in B$ is defined by $\Sp(x)=\{c\in \comp: x-cI \text{ invertible }\}$. The spectrum is a nonempty closed and bounded set and hence compact. 
An element $x$ is self-adjoint if $x=x^*$, normal if $x^*x=xx^*$ and positive (strictly positive) if $x$ is self-adjoint and $\Sp(x)\subset [0,\infty)((0,\infty))$. A self-adjoint
element has a real spectrum and conversely. Since $x=x_1+ix_2$ with $x_1=(x+x^*)/2$ and $x_1=(x+x^*)/2i$ any element of a $\cstar$ algebra can be decomposed into self-adjoint ``real'' and ``imaginary'' parts. 
The positive elements define a partial order on $A$: 
$x\leq y$ iff $y-x\geq 0$ (positive). A positive element $a$ has a unique square-root $\sqrt{a}$ such that $\sqrt{a}\geq 0\text{ and } (\sqrt{a})^2=a$. If $x$ is self-adjoint, $x^2\geq 0$ and $|x|=\sqrt{x^2}$.  A self-adjoint element $x$ has a decomposition $x=x_+-x_-$ into positive and negative parts where \(x_+=(|x|+x)/2\ \text{ and } x_-=(|x|-x)/2)\) are positive. An element $p\in B$ is a projection 
if $p$ is self-adjoint and $p^2=p$. Given two $\cstar$-algebras $A$ and $B$ a homomorphism $F$ is a linear map preserving the product and $^*$ structures. 
A homomorphism is positive if it maps positive elements to positive elements. A (linear) functional on $A$ is a linear map $A\rightarrow \comp$. A positive functional $\omega$ such that $\omega(\unit)=1$ is called a {\em state}. The set of states $G$ is convex. The extreme points 
are called {\em pure states} and $G$ is the convex closure of pure states (Krein-Millman theorem). A set $B\subset A$ is called a subalgebra if it is a 
$\cstar$ algebra with the inherited product. A subalgebra is called unital if it contains the identity of $A$.
Our primary interest will be on {\em abelian} or commutative algebras.  
The basic representation theorem (Gelfand-Naimark) \cite{KR1} states that: {\em An abelian $\cstar$ algebra with unity is isomorphic to the algebra $C(X)$ continuous complex-valued functions on a compact Hausdorff space $X$}. 

Now let $X=\{a_1, \dotsc, a_n\}$  be a finite set with discreet topology. Then $A=C(X)$ is the set of all functions $X\rightarrow \comp$. The algebra $C(X)$ can be considered as the algebra of (complex) random variables on the finite probability space $X$. Let $x_i(a_j)=\delta_{ij},\; i,j=1,\dotsc, n$. Here $\delta_{ij}=1 \text{ if }i=j \text{ and } 0$ otherwise. The functions $x_i\in A$ form a basis for $A$. Their multiplication table is particularly simple: $x_ix_j=\delta_{ij}x_i$. They also satisfy $\sum_i x_i=\unit$. These are projections in $A$. They are orthogonal in the sense that $x_ix_j=0\text{ for }i\neq j$. We call any basis consisting of elements of norm 1 with distinct elements orthogonal {\em atomic}. A set of linearly independent elements $\{y_i\}$ satisfying $\sum_i y_i=\unit$ is said to be complete. The next theorem gives us the general structure of any finite-dimensional algebra.  
\begin{thm}\label{thm:structFinite}
Let $A$ be a finite-dimensional abelian $\cstar$ algebra. Then there is a unique (up to permutations) complete atomic basis $\cali{B}=\{x_1, \dotsc, x_n\}$. That is, the basis elements satisfy
\beq \label{eq:structFinite}
x_i^*=x_i,\; x_ix_j=\delta_{ij}x_i,\;  \norm{x_i}=1 \text{ and }\sum_i x_i =\unit,\;
\eeq
Let $x=\sum_i a_ix_i\in A$. Then $\Sp(x)=\{a_i\}$ and hence $\norm{x}=\max_i\{|a_i|\}$. 
\end{thm}
We next describe an important construction for $\cstar$ algebras. Given two $\cstar$ algebras $A$ and $B$, the tensor product $A\tensor B$ is defined as follows. As a set it consists of all finite linear combinations of symbols of the form $\{x\tensor y:x\in A,y\in B\}$ subject to the conditions that the map $(x,y)\rightarrow x\tensor y$ is bilinear in each variable. Hence, if $\{x_i\} \text{ and } \{y_j\}$ are bases for $A$ and $B$  respectively then $\{x_i\tensor y_j\}$ is a basis for $A\tensor B$. The linear space $A\tensor B$ becomes an algebra by defining 
$(x\tensor y)(u\tensor z)=xu\tensor yz$ and extending by bilinearity. 
The $*$ is defined by $(x\tensor y)^*=x^*\tensor y^*$ and extending {\em anti-linearly}. We will define the norm in a more general setting. Our basic model will be an {\em infinite} tensor product of finite dimensional $\cstar$ algebras which we present next.

\newcommand{\inftens}[1]{\bigotimes^{\infty}{#1}}
Let $A_k,\;k=1,2,\dotsc,$ be finite dimensional abelian $\cstar$ algebras with atomic basis $B_k=\{x_{k1},\dotsc,x_{kn_k}\}$. Let $B^{\infty}$ be the set consisting of all infinite strings of the form \(z_{i_1}\tensor z_{i_2}\tensor \cdots\) where all but a finite number ($>0$) of $z_{i_k}$s are equal to $\unit$ and if some $z_{i_k}\neq \unit$ then $z_{i_k}\in B_k$. 
\commentout{
Explicitly, $B^{\infty}$ consists of strings of the form $z_{i_1}\tensor z_{i_2}\tensor \cdots \tensor z_{i_k}\tensor\unit\tensor\unit\tensor\cdots,\;k=1,2,\dotsc $ and $z_i\in B$.} Let $\tilde{\mathfrak{A}}=\tensor_{i=1}^{\infty}{A}_i $ be the vector space with basis $B^{\infty}$ such that $z_{i_1}\tensor z_{i_2}\tensor \cdots \tensor z_{i_k}\tensor\cdots $ is linear in each factor separately. 
\commentout{  
\[
\begin{split}
&z_{1_1}\tensor\cdots \tensor (az_{i_k}+bz'_{i_k})\tensor z_{i_{k+1}}\tensor \cdots= \\
&a(z_{1_1}\tensor \cdots \tensor z_{i_k}\tensor z_{i_{k+1}}\tensor \cdots)+ 
b(z_{1_1}\tensor \cdots \tensor z'_{i_k}\tensor z_{i_{k+1}}\tensor \cdots).\\
\end{split}
\]

Clearly every $\alpha \in \tilde{\mathfrak{A}} $ is a finite linear combination of elements in $B^{\infty}$. }
We define a product in $\tilde{\mathfrak{A}}$ as follows. First, for elements of $B^{\infty}$: 
\((z_{i_1}\tensor z_{i_2}\tensor\cdots )(z'_{i_1}\tensor z'_{i_2}\tensor\cdots )=(z_{i_1}z'_{i_1}\tensor z_{i_2}z'_{i_2}\tensor\cdots )\)
We extend the product to whole of $\tilde{\mathfrak{A}}$ by linearity. Next define a norm  by:  
\(\norm{\sum_{i_1,i_2,\dotsc}a_{i_1i_2\cdots}z_{i_1}\tensor z_{i_2}\tensor \cdots }=\sup\{|a_{i_1i_2\cdots}|\}\).  $B^{\infty}$ is an atomic basis. It follows that $\tilde{\mathfrak{A}}$ is an abelian normed algebra. We define $*$-operation by 
\(\left(\sum_{i_1,i_2,\dotsc}a_{i_1i_2\cdots}z_{i_1}\tensor z_{i_2}\tensor \cdots \right)^*=\sum_{i_1,i_2,\dotsc}\conj{a_{i_1i_2\cdots}}z_{i_1}\tensor z_{i_2}\tensor \cdots \)
It follows that for $x\in \tilde{\mathfrak{A}}$, $\norm{xx^*}=\norm{x}^2$. Finally, we complete the norm \cite{KR1} and call the resulting $\cstar$ algebra $\mathfrak{A}$. With these definitions  $\mathfrak{A}$ is a $\cstar$ algebra. We call a $\cstar$ algebra $B$ of {\bf finite type} if it is either finite dimensional or infinite tensor product of finite-dimensional algebras. An important special case is when all the factor algebras $A_i=A$. We then write the infinite tensor product $\cstar$ algebra as $\inftens{A}$. Intuitively, the elements of an atomic basis $B^{\infty}$ of $\inftens{A}$ correspond to strings from an alphabet (represented by the basis $B$). Of particular interest is the 2-dimensional algebra $D$ corresponding to a binary alphabet. 
\commentout{
It can be shown that there are injective algebra maps---\(\cali{J}: \inftens{G}\rightarrow \inftens{A} \text{ and } \cali{J}': \inftens{A} \rightarrow \inftens{G}\). This is relevant for coding theory. }

The next step is to describe the state space. Given a $\cstar$ subalgebra $V\subset A$ the set of states of $V$ will be denoted by $\mathscr{S}(V)$. Let $\mathfrak{A}=\tensor^{\infty}_{i=1} A_i$ denote the infinite tensor product of finite-dimensional algebras $A_i$. An infinite product state of $\mathfrak{A}$ is a functional of the form
\( \Omega=\omega_1\tensor \omega_2\tensor\cdots \text{ such that }\omega_i\in \mathscr{S}(A_i)\)
This is indeed a state of $\mathfrak{A}$ for if $\alpha_k = z_1\tensor z_2 \tensor \cdots \tensor z_k\tensor\unit\tensor\unit\cdots\in \mathfrak{A}$ then 
\(\Omega(\alpha)=\omega_1(z_1)\omega_2(z_2)\cdots \omega_k(z_k), \)
a {\em finite} product. 
\commentout{
Since an arbitrary element of $\mathfrak{A}$ is the limit of sequence of finite sums of elements of the form $\alpha_k,\; k=1,2,\dotsc $ $\Omega$ is bounded by the principle of uniform boundedness. Clearly, it is positive.}
A general state on $\mathfrak{A}$ is a convex combination of product states like $\Omega$. 
Finally, we discuss another useful construction in a $\cstar$ algebra $A$. If $f(z)$ is an analytic function whose Taylor series $\sum_{n=0}^{\infty}a_n (z-c)^n$  converges in a region $|z-c|<R$. Then the series $\sum _{n=0}^{\infty}(x-c\unit)^n$ converges and it makes sense to talk of analytic functions on a $\cstar$ algebra. If we have an atomic basis $\{x_1,x_2,\dotsc \}$ in an abelian $\cstar$ algebra then the functions are particularly simple in this basis. Thus if $x=\sum_i a_ix_i$ then $f(x)=\sum_i f(a_i)x_i$ provided that $f(a_i)$ are defined in an appropriate domain. 

We gave a brief description of $\cstar$ algebras. We now
 introduce an algebraic model of probability which is used later to model communication processes. 
In this model we treat random variables as elements of a $\cstar$ algebra. The probabilities are introduced via states. 
\commentout{
We emphasize again that random variables often  represent quantities that are actually measured or observed- the voltage across a resistor, the currents in an antenna, the position of a Brownian particle and so on. The probability distribution corresponds to the {\em state} of the devices that produce these outputs. We will take the alternative view and start with these observables as our basic objects.}
\newcommand{\scrp}[1]{{\mathscr #1}}
\newcommand{\intd}{\mathrm{d}}
A classical observable algebra is a complex abelian $\cstar$ algebra $A$. We can restrict our attention to real algebras whenever necessary. The Riesz representation theorem \cite{Rudin} makes it possible identify $\omega$ with some {\em probability measure}. 
A {\em probability algebra} is a pair $(A, S)$ where $A$ is an observable algebra and $S\subset \scrp{S}(A)$ is a set of states. A probability algebra is defined to be {\em fixed} if $S$ contains only one state. 

\noindent
Let $\omega$ be a state on an abelian $\cstar$ algebra $A$. Call two elements $x,y\in A$ {\em uncorrelated in the state} $\omega$ if $\omega(xy)=\omega(x)\omega(y)$. This definition depends on the state: two uncorrelated elements can be correlated in some other state $\omega'$. A state $\omega$ is called multiplicative if $\omega(xy)=\omega(x)\omega(y)$ for all $x,y\in A$.  The set of states, $\mathscr{S}$, is convex. The extreme points of $\mathscr{S}$ are called {\em pure} states. In the case of abelian $\cstar$ algebras a state is pure if and only of it is multiplicative \cite{KR1}. Thus, in a pure state any two observables are uncorrelated. This is not generally true in the non-abelian quantum case. Now  we can introduce the important notion of {\em independence}. Given $S\subset A$ let $A(S)$ denote the subalgebra generated by $S$ (the smallest subalgebra of $A$ containing $S$). Two subsets $S_1,S_2\subset A$ are defined to be {\em independent} if all the pairs $\{(x_1,x_2): x_1\in A(S_1), x_2\in A(S_2)\}$ are uncorrelated. As independence and correlation depend on the state we sometimes write $\omega$-independent/uncorrelated. Independence is a much stronger condition than being uncorrelated. 
\commentout{
However, in 2 dimensions $x\text{ and }x'$ are uncorrelated if and only if one of them is 0 or $c\unit$. Let us note that as in the quantum case two dimensions is an exceptional case.}
 The next theorem states the structural implications of independence. 
\begin{thm} \label{thm:structIndep}
Two sets of observables $S_1,S_2$ in a finite dimensional abelian $\cstar$ algebra $A$ are independent in a state $\omega$ if and only if for the subalgebras $A(S_1)$ and $A(S_2)$ generated by $S_1$ and $S_2$ respectively there exist states $\omega_1 \in \scrp{S}(A(S_1)),\; \omega_2\in \scrp{S}(A(S_2))$ such that $(A(S_1)\tensor A(S_2),\{\omega_1\tensor\omega_2\})$ is a cover of $(A(S_1S_2),\omega')$ where $A(S_1S_2)$ is the subalgebra generated by $\{S_1,S_2\}$ and $\omega'$ is the restriction of $\omega$ to $A(S_1S_2)$. 
\end{thm}
\commentout{
The next step is to extend the notion of independence to more than two subsets. Let $S_1,\dotsc,S_k\subset A$ and $\omega$ a state of $A$. Then the subsets are defined to be $\omega$-independent if for all $x_i\in A(S_i),\; i=1,\dotsc, k$ we have 
\[ \omega(x_1\cdots x_k)=\omega(x_1)\cdots \omega(x_k)\]
Here $A(S_i)$ is the subalgebra generated by $S_i$. We can then show that for states $\omega_i\in \scrp(A(S_i))$, the restriction of $\omega$ to $A(S_i)$ the pair 
\((A(S_1)\tensor\cdots \tensor A(S_k), \omega_1\tensor\cdots\tensor \omega_k)\) is a cover of $A(S_1\dotsc S_k),\omega'$, where $\omega'$ is the restriction of $\omega$ to $A(S_1\dotsc S_k)$, the algebra generated by $S_1,\dotsc, S_k$.}
We thus see the relation between independence and (tensor) product states in the classical theory. Next we show how one can formulate another important concept, {\em distribution function} ({\bf d.f}) in the algebraic framework. We restrict our analysis to $\cstar$ algebras of finite type. The general case is more delicate and is defined using approximate identities in subalgebras in \cite{Patra09}. The idea is that we approximate indicator functions of sets by a sequence of elements in the algebra. In the case of finite type algebras the sequence converges to a projection operator $J_S$. 
\commentout{
\noindent
{\bf Definition.} Let $S=\{x_1,x_2,\dotsc, x_n\}$ be a finite self-adjoint subset of $A$ where $(A,\omega)$ is a fixed probability algebra of finite type. For ${\tt t}=(t_1,t_2,\dotsc, t_n )\in \real$ let $S_{\tt t}\subset A$ denote the set of elements  $\{(t_i\unit - x_i):i=1,\dots, n\}$. Let $J_S$ be the identity in the (annihilator) subalgebra  $(S_{\tt t})_a\equiv \{ x\in A : xs=0\forall s\in S_{\tt t}\}$. Then the $\omega$-distribution of $S$ is defined to be the real function \( f_S({\tt t})= \omega(J_S) \). Using the d.f we can define the cumulative distribution function $F_S({\tt t})= \sum _{q_i} f_S({\tt t}: q_i\leq t_i)$. The sum is well defined since $f_S$ has only finitely many nonzero values. 

\noindent
We explain the rationale of this definition.} Thus, if we consider a representation where the elements of $A$ are functions on some finite set $F$ then $J_S$ is precisely the indicator function of the set $S'=\{c:x_i(c)-t_i=0:c\in F\text{ and }i=1,\dotsc, n\}$. The set $S'$ corresponds to the subalgebra $(S_{\tt t})_a$ and $J_S$, a projection in $A$, acts as identity in $(S_{\tt t})_a$. 
\commentout{
The following shows  the existence of $J_S$ by an explicit construction. Suppose $S$ contains a single element $x$. Writing $x=\sum_ia_ix_i$ in some atomic basis we have $x-t=\sum a_jP_j$ where $a_j\neq 0$ are distinct and $P_j$ are projections. Now use Lagrange interpolation to obtain polynomials $g_j$ such that $g_j(a_k)=\delta_{jk}$ and $g_j(0)=0$. Then $Q=\sum_j P_j$ is the required projection. If there are more elements in $S$ then let $Q_j$ be the projection for $x_j-t_j$ and $Q$ the product of $Q_j$'s. Then $J_s=\unit-Q$. Note that we have an explicit formula for distribution function that we can use it to prove properties of d.f. }
From the notion of distribution functions we can define now probabilities $Pr(a\leq x\leq b)$ in the algebraic context. We can now formulate problems in any discrete stochastic process in finite dimensions. The algebraic method actually provides practical tools besides theoretical insights as the example of ``waiting time'' shows \cite{Patra09}. 
%
Now we consider the algebraic formulation of a basic limit theorem of probability theory: the {\bf weak law of large numbers}. From information theory perspective it is perhaps the most useful limit theorem. Let $X_1,X_2,\cdots, X_n$ be independent, identically distributed (i.i.d) bounded random variables on a probability space $\Omega$ with probability measure $P$. Let $\mu$ be the mean of $X_1$. Recall the 
{\em Weak law of large numbers}. Given $\epsilon>0$ 
\[\lim_{n\rightarrow \infty}P(|S_n=\frac{X_1+\cdots +X_N}{n}-\mu |>\epsilon)=0\]
We have an algebraic version of this important result. 
\begin{thm} [Law of large numbers (weak)] \label{thm:weak-law} 
If $x_1,\dotsc,x_n,\dotsc$ are  
$\omega$-\\independent self-adjoint elements in an observable algebra and $\omega(x_i^k)=\omega(x_j^k)$ for all positive integers $i,j\text{ and }k$ (identically distributed) then 
\[\lim_{n\rightarrow \infty} \omega(|\frac{x_1+\dotsb+x_n}{n}-\mu|^k)=0 \text{ where } \mu=\omega(x_1) \text{ and } k>0\]
\end{thm}
\commentout{
\begin{proof}
We may assume $\mu=0$ (by reasoning with $x_i-\omega(x_i)$ instead of $x_i$). First we prove the statement for $k=2$. Then $\omega(|\frac{x_1+\dotsb+x_n}{n}|)^2=\sum_i\omega(x_i^2)/n^2=\omega(x_1^2)/n$. The first equality follows from independence ($\omega(x_ix_j)=\omega(x_i)\omega(x_j)=0\text{ for }i\neq j$) the second from the fact that they are identically distributed. The case $k=2$ is now trivial. Now let $k=2m$. Then $|x_1+\dotsb +x_n|^k=(x_1+\dotsb +x_n)^k$. Put $s_n=(x_1+\dotsb +x_n)/n$. Expanding $s_n^k$ in a multinomial series we note that  independence and the fact that $\omega(x_i)=0$ implies that all terms in which at least one of the $x_i$ has power 1 do not contribute to $\omega(s_n^k)$. The total number of the remaining terms is $O(n^{m})$. Since the denominator is $n^{2m}$ we see that $\omega(s_n^k)\rightarrow 0$. Since for any $x\in A$, $|x|=(x^2)^{1/2}$ can be approximated by polynomials in $x^2$ we conclude that $\omega(|s_n|)\rightarrow 0$. Finally, using the Cauchy-Schwartz type inequality $\omega(|s_n|^{2r+1})\leq \omega(s_n^2)\omega(s_n^{2r})$ we see that the theorem is true for all $k$. 
\end{proof}
}
Using the algebraic version of Chebysev inequality the above result implies the following. 
Let $x_1,\dotsc, x_n \text{ and } \mu$ be as in the Theorem and set $s_n=(x_1+\dotsb+x_n)/n$. Then for any $\epsilon >0$ there exist $n_0$ such that for all $n>n_0$ 
\(
P(|s_n-\mu|>\epsilon) <\epsilon 
\)
\commentout{
\begin{proof}
Using Chebysev inequality we have \( P(|s_n-\mu|>\epsilon)= P(|s_n-\omega(s_n)|>\epsilon) \leq \frac{\omega(|s_n-\mu|^2)}{\epsilon^2}\). As $\omega(|s_n-\mu|^2)\rightarrow 0$ (Theorem \ref{thm:weak-law}) there is $n_0$ such that $\omega(|s_n-\mu|^2)<\epsilon^3$ for $n>n_0$. 
\end{proof}
}
\section{Communication and Information}\label{sec:info}
We now come to our original theme: an algebraic framework for communication and information processes. Since our primary goal is the modeling of information processes we refer to the simple model of communication in the Introduction and model different aspects of it. In this work we will only deal with sources with a finite alphabet. 

\noindent
{\bf Definition.} {\em A source  is a pair $\scrp{S}=(B,\Omega)$ where $B$ is an atomic basis of a finite-dimensional abelian $\cstar$ algebra $A$ and $\Omega$ is a state in $\inftens{A}$}.
 
\noindent
This definition abstracts the essential properties of a source. The basis $B$ is called the {\em alphabet}. A typical output of the source is of the form $x_1\tensor x_2\tensor \dotsb\tensor x_k\tensor \unit \tensor \dotsb \in B^{\infty}$, the infinite product basis of $\inftens{A}$. We identify \(\hat{x}_k= \unit\tensor\dotsb\tensor\unit\tensor x_k\tensor\unit\tensor\dotsb\) with the $k$th signal. If these are independent 
then Theorem \ref{thm:structIndep} tells us that $\Omega$ must be product state. Further, if the state of the source does not change then $\Omega=\omega\tensor\omega\tensor\dotsb$ where $\omega$ is a state in $A$. 
For a such state $\omega$ define: 
\(\cali{O}_{\omega}=\sum_{i=1}^n \omega(x_i)x_i, \; \{x_1, \dotsc, x_n\},\; x_i\in B\)
We say that $\cali{O}_{\omega}$ is the ``instantaneous'' output of the source in state $\omega$. 
Let $A'$ be another finite-dimensional $\cstar$ algebra with atomic basis $B'$
 A source coding is a linear map $f:B\rightarrow  T= \sum_{k= 1}^m \tensor^k A' $. Such that for $x\in B$, $f(x)=x'_{i_1}\tensor x'_{i_2}\tensor \dotsb\tensor x'_{i_r},\; r\leq k$ with $x'_{i_j}\in B'$. Thus each ``letter'' in the alphabet $B$ is coded by ``words'' of maximum length $k$ from $B'$. 

\commentout{
Let us consider an example to clarify these points. 
Let $\{x_0,x_1,x_2,x_3\}$ be an atomic basis for $B$. Let $B'=G$ with atomic basis $\{y_0,y_1\}$. Define $f_1$ by $f_1(x_0)=y_0,f_1(x_1)= y_1, f_1(x_2)=y_0\tensor y_1 \text{ and } f_1(x_3)=y_1\tensor y_0$. Denote by $\hat{f}_1$ its extension to tensor products. Since $\hat{f}_1(x_0\tensor x_1)=y_0\tensor y_1=\hat{f}_1(x_2)$, $\hat{f}_1$ is not injective. Hence it cannot be inverted on its range. Consider next the map $f_2(x_0)=y_0,f_2(x_1)=y_0\tensor y_1,f_2(x_2)=y_0\tensor y_1\tensor y_1\text{ and }f_2(x_3)=y_1\tensor y_1\tensor y_1$. This map is invertible but one has to look at the complete product before finding the inverse. It is not {\em prefix-free}.}
A code $f:B\rightarrow T$ is defined to be prefix-free if for distinct members  $x_1,x_2$ in an atomic basis of $B$, $f'(x_1)f'(x_2)=0$ where $f'$ is the map $f': B\rightarrow \inftens B'$ induced by $f$. That is, distinct elements of the atomic basis of $B$ are mapped to {\em orthogonal} elements. Thus the
``code-word'' $ z_1\tensor z_1\tensor\dotsb \tensor z_k \tensor \unit\tensor \unit\tensor\dotsb$ is not orthogonal to another $ z'_1\tensor z'_1\tensor\dotsb \tensor z'_m \tensor \unit\tensor \unit\tensor\dotsb$ with $k\leq m$ if and only if $z_1=z'_1,\dotsc, z_k=z'_k$. 
\commentout{
We observe that one has to be careful about correspondence between the two approaches. For example, one might be tempted to identify the identity $\unit$ with the empty string but the $\unit$ is the sum of the members of an atomic basis! The binary operation ``+'' has a relatively lesser role in the classical formalism but it is crucial in the quantum framework (via superposition principle).}
The useful Kraft inequality can be proved using algebraic techniques. 
Corresponding to a finite sequence $k_1\leq k_2\leq \dotsb \leq k_m$ of positive integers let $\alpha_1,\dotsc, \alpha_m$ be a set of prefix-free elements in $\sum_{i\geq 1} \tensor^i A'$ such that $\alpha_i\in \tensor^{k_i} A'$. Further, suppose that each $\alpha_i$ is a tensor product of elements from $B'$. Then 
\beq \label{eq:Kraft}
 \sum_{i=1}^m n^{k_m-k_i} \leq n^{k_m}
\eeq
This inequality is proved by looking at bounds on dimensions of a sequence of orthogonal subspaces.   
\commentout{
The Kraft inequality is valid for decipherable sequences \cite{McMillan}. However, the proof is essentially combinatorial. The Kraft inequality also provides a sufficiency condition for prefix-free code \cite{Ciszar,CoverT}. Thus the existence of a decipherable code of word-lengths $(k_1,k_2,\dotsc,k_m)$ implies the existence of a prefix-free code of same word-lengths.}
In the following, we restrict ourselves to prefix-free codes. 
\commentout{
\begin{lem}
Let $f$ be a continuous real function on $(0,\infty )$ such that $xf(x)$ is convex and $\lim_{x\rightarrow 0} xf(x)=0$. Let $A$ be a finite-dimensional $\cstar$ algebra with atomic basis $\{x_1,\dotsc, x_n\}$ and $\omega$ a state on $A$. Then for any set of numbers $\{a_i\; :i=1,\dotsc, n;\; a_i>0\text{ and } \sum_i a_i\leq 1\}$ we have 
\(\omega(\sum_i f(\frac{\omega(x_i)}{a_i})x_i) \geq f(1)\)
\end{lem}

\begin{proof}
Let $\omega(x_i)=p_i$. We have to show that $\sum p_if(p_i/a_i) \geq f(1)$. First assume that all $p_i>0$ and $\sum_i a_i=1$. Then 
\[ \sum_i p_i f(p_i/a_i) =\sum_i a_i\frac{p_i }{a_i}f(\frac{p_i }{a_i}) \geq f(\sum p_i)=f(1)\]
by convexity of $xf(x)$. The general case can be proved by starting with $a_i$ corresponding to $p_i>0$ and adding extra $a_j$'s to satisfy $\sum_i a_i=1$ if necessary. The corresponding $p_j$ is set to $0$. Now define a new function $g(x)=xf(x), \; x>0$ and $g(0)=0$. The conclusion of the lemma follows by arguing as above with $g$.   
\end{proof}
}
Using convexity function $f(x)=-\log {x}$ and the Kraft inequality \ref{eq:Kraft} we deduce the following. 
\begin{propn}[Noiseless coding]
Let $\mathscr{S}$ be a source with output $\cali{O}_{\omega}\in A$, a finite-dimensional $\cstar$ algebra with atomic basis $\{x_1,\dotsc, x_n\}$ (the alphabet). Let $g$ be prefix-free code such that $g(x_i)$ is a tensor product of $k_i$ members of the code basis. Then 
\(\omega(\sum_i k_ix_i+\log{\cali{O}_{\omega}})\geq 0\)
\end{propn}

Next we give a simple application of the law of large numbers. First define a positive functional $\tr$ on a finite dimensional abelian $\cstar$ algebra $A$ with an atomic basis $\{x_1,\dotsc, x_d\}$ by 
\(\tr =\omega_1+\dotsb+\omega_d\) where $\omega_i$ are the dual functionals. It is clear that $\tr$ is independent of the choice of atomic basis.  
\begin{thm}[Asymptotic Equipartition Property (AEP)]\label{thm:AEP}
Let $\scrp{S}$ be a source with output $\cali{O}_{\omega}=\sum_{i=1}^d \omega(x_i)x_i$ where $\omega$ is a state on the finite dimensional algebra with atomic basis $\{x_i\}$. Then given $\epsilon >0$ there is a positive integer $n_0$ such that for all $n>n_0$ 
\[ P(2^{n(H(\omega)-\epsilon)}\leq \tensor^n \cali{O}_{\omega}\leq 2^{n(H(\omega)+\epsilon)}) > 1- \epsilon \]
where $H=\omega(\log_2(\cali{O}_{\omega}))$ is the {\em entropy} of the source and the probability distribution is calculated with respect to the state $\Omega_n=\omega\tensor\dotsm\tensor\omega$ ($n$ factors) of $\tensor^n A$. If $Q$ denotes the identity in the subalgebra generated by $(\epsilon I-|\log_2(\tensor^n\cali{O}_\omega)+nH|)_+$ then 
\[(1-\epsilon)2^{n(H(\omega)-\epsilon)}\leq \tr(Q) \leq 2^{n(H(\omega)+\epsilon)} \]
\end{thm}
\commentout{
The function $\log{x}\equiv \log_2{x}\;(=\ln{x}/\ln{2})$ is defined for strictly positive elements of a $\cstar$ algebra. We extend the definition to all non-zero $x\geq 0$. Let $\{y_i\}$ be a  atomic basis in an abelian $\cstar$ algebra. Let $y=\sum_i a_iy_i$ with $a_i\geq 0$. Then define $\log_2{y}=\sum_i b_i y_i$ where $b_i=\log{a_i}$ if $a_i>0$ and 0 otherwise. This definition implies that some standard properties of $\log$ are no longer true (e.g.\ $2^{\log{x}}\neq x$). But in the present context it gives the correct result when we take expectation values as in the formulas in the theorem. 

A somewhat longer but mathematically better justified route is to ``renormalize'' the state. Thus if $\omega(x_i)=0$ for $k$ indices we define $\omega'(x_i)=\delta$ where $\delta$ is arbitrarily small but positive and $\omega'(x_j)=\omega(x_j)-k\delta$ where $\omega'(x_j)>k\delta$. If we can prove the theorem now for $\omega'$ and since the relations are valid in the limit $\delta\rightarrow 0$ then we are done. We will not take this path but implicitly assume that the probabilities are positive.} Note that the element $Q$ is a projection on the subalgebra generated by $(\epsilon I-|\log_2(\tensor^n\cali{O}_\omega)-nH|)_+$. It corresponds to the set of strings whose probabilities are between $2^{-nH-\epsilon}$ and $2^{-nH+\epsilon}$. The integer $\tr(Q)$ is simply the cardinality of this set. 
\commentout{
\begin{proof}[Proof of the theorem]
First note that $\log{ab}=\log a+\log b$ for elements $a,b\geq 0$ in $A$. We can write $\tensor^n\cali{O}_{\omega}=X_1X_2\dotsb X_n$ where $X_i=\unit\tensor\unit\tensor\dotsm\tensor\cali{O}_{\omega}\tensor\unit\tensor\dotsm\tensor\unit$ with $\log{\cali{O}_{\omega}}$ in the $i$th place. The fact that $\Omega_n$ is a product state on $\tensor^n A$ (corresponding to a source whose successive outputs are independent) implies that $X_i$ are independent and identically distributed. We can now apply the corollary to Theorem \ref{thm:weak-law} yielding 
\(P(|\log{(\tensor^n\cali{O}_{\omega})} -\Omega_n(\log{X_1})|>\epsilon)=P(|\log{(\tensor^n\cali{O}_{\omega})} -\omega(\log{(\cali{O}_{\omega})})|>\epsilon)\). 
\end{proof}
}

We now come to the most important part of the communication model: the {\em channel}. 
The original paper of Shannon characterized channels by a transition probability function. We will consider only (discrete) memoryless channel (DMS). A  DMS channel has an input alphabet $X$ and output alphabet $Y$ and 
\commentout{
a sequence of random functions $\phi_n: X^n\rightarrow Y^n$. The latter are characterized by probability distributions $p_n(y^{(n)}|x^{(n)})$, the interpretation being: $\phi_n(x^{(n)})=y^{(n)}$ with conditional probability  $p_n(y^{(n)}|x^{(n)})$. Note that the distribution depends on the entire history. We say that such a channel has (infinite) memory. A channel has finite memory if there is an integer $k\geq 0$ such that if $x^{(n)}=x_nx_{n-1}\cdots x_{n-k+1}\dotsc x_1$ then $p_n(y^{(n)}|x^{(n)})= p_n(y^{(n)}|x'^{(n)})$ for any string $x_n'$ of length $n$ such that $x_n'=x_n, \dotsc, x_{n-k+1}'=x_{n-k+1}$. That is, the probability distribution depends on the most recent $k$ symbols seen by the channel. A channel is {\em memoryless} if $k=1$. Since we will be dealing mostly with discrete memoryless channels (DMS) this property will be tacitly assumed unless stated otherwise. In the memoryless case it is easy to show the simple form of transition probabilities 
\beq
\begin{split}
&p_n(y^{(n)}|x^{(n)})=p_n(y_1\dotsc y_n|x_1\dotsc x_n)\\
& =p(y_1|x_1)p(y_2|x_2)\dotsb p(y_n|x_n)\\ 
\end{split}
\eeq
This motivates us to define the}
a {\em channel transformation matrix} $C(y_j|x_i)$ with $y_j\in Y$ and $x_i\in X$. Since the matrix $C(y_j|x_i)$  represents the probability that the channel outputs $y_j$ on input $x_i$ we  have $\sum_j C(y_j|x_i)=1$ for all $i$: $C(ij)=C(y_j|x_i)$ is {\em row stochastic}. This is the standard formulation.  \cite{Ciszar,CoverT}. We now turn to the algebraic formulation. 

\noindent
{\bf Definition.} A DMS channel $\cali{C}=\{X,Y,C\}$ where $X$ and $Y$ are abelian $\cstar$ algebras of dimension $m$ and $n$ respectively and $C: Y \rightarrow X$ is a unital positive map. The algebras $X$ and $Y$ will be called the input and output algebras of the channel respectively. Given a state $\omega$ on $X$ we say that $(X,\omega)$ is the input source for the channel. 

\noindent
Sometimes we write the entries of $C$ in the more suggestive form $C_{ij}=C(y_j|x_i)$ where $\{y_j\}$ and $\{x_i\}$ are atomic bases for $Y$ and $X$ respectively. Thus $C(y_j)=\sum_i C_{ij}x_i= \sum_i C(y_j|x_i)x_i$. Note that in our notation $C$ is an $m\times n$ matrix. Its transpose $C^T_{ji}=C(y_j|x_i)$ is the channel matrix in the standard formulation. We have to deal with the transpose because the channel is a map {\em from} the output alphabet to the input alphabet. This may be counterintuitive but observe that any map $Y\rightarrow X$ defines a unique dual map $\cali{S}(X)\rightarrow \cali{S}(Y)$, on the respective state spaces. Informally, a channel transforms a probability distribution on the input alphabet to a distribution on the output. 
\commentout{
Let us note that in case of abelian algebras every positive map is guaranteed to be {\em completely positive} \cite{Tak1}. This is no longer true in the non-abelian case. Hence for the quantum case completely positivity has to be explicitly imposed on (quantum) channels.}
We characterize a channel by input/output algebras (of observables) and a positive map. Like the source output we now define a useful quantity called {\em channel output}. Corresponding to the atomic basis $\{y_i\}$ of $Y$ let $\tensor^k y_{i(k)}$ be an atomic basis in $\tensor^n Y$. Here $i(k)=(i_1i_2\dotsc i_k)$ is a multi-index. Similarly we have an atomic basis $\{\tensor^k x_{j(k)}\}$ for $\tensor^k X$. The level-$k$ channel output is defined to be 
\(O^k_C = \sum_{i(k)} y_{i(k)}\tensor C^{(k)}(y_{i(k)})  \).
Here $C^{(k)}$ represents the channel transition probability matrix on the $k$-fold tensor product corresponding to strings of length $k$. In the DMS case it is simply the $k$-fold tensor product of the matrix $C$. The channel output defined here encodes most important features of the communication process. First, given the input source function  \(\cali{I}_{\omega^k}=\sum_i\omega^k(x_{i(k)})x_{i(k)}\) the output source function is defined by 
\(\cali{O}_{\tilde{\omega}^k} = I\tensor \tr_{\tensor^kX}((\unit \tensor \cali{I}_{\omega^k})O^k_c)
=\sum_i \sum_j C(y_{i(k)}|x_{j(k)})\omega^k(x_{j(k)})y_{i(k)}\). 
Here, the state $\tilde{\omega}^k$ on the output space $\tensor^k Y$ can be obtained via the dual $\tilde{\omega}^k(y)=\tilde{C}^k(\omega^k)(y)=\omega^k(C^k(y))$. The formula above is an alternative representation which is very similar to the quantum case. The {\em joint output} of the channel can be considered as the combined output of the two terminals of the channel. Thus the joint output
\beq \label{eq:chOutputJoint}
\begin{split}
&\cali{J}_{\tilde{\Omega}^k} = (\unit \tensor \cali{I}_{\omega^k})O^k_C=\sum_{ij} \Omega^k(y_{i(k)}\tensor x_{j(k)})y_{i(k)}\tensor x_{j(k)},  \\
& \Omega^k(y_{i(k)}\tensor x_{j(k)})\equiv C(y_{i(k)}|x_{j(k)})\omega(x_{j(k)})
\end{split}
\eeq
Let us analyze the algebraic definition of channel given above. For simplicity of notation, we restrict ourselves to level 1. The explicit representation of channel output is 
\(\sum_i y_i\tensor \sum_j C(y_i|x_j)x_j \)
We interpreted this as follows: if on the channel's out-terminal $y_i$ is observed then the input could be $x_j$ with probability \(C(y_i|x_j)\omega(x_j)/\sum_jC(y_i|x_j)\omega(x_j)\). Now suppose that for a fixed $i$  $C(y_i|x_j)=0$ for all $j$ except one say, $j_i$. Then on observing $y_i$ at the output we are certain that the the input is $x_{j_i}$. If this is true for all values of $y$ then we have an instance of a lossless channel. Given $1\leq j\leq n$ let $d_j$ be the set of integers $i$ for which $C(y_i|x_j)> 0$. If the channel is lossless then $\{d_j\}$ form a partition of the set $\{1,\dotsc, m\}$. The corresponding channel output is 
\( O_C= \sum_j \Bigl(\sum_{i\in d_j} C(y_i|x_j)y_i\Bigr)\tensor x_j\). 
At the other extreme is the {\em useless} channel in which there is no correlation between the input and the output. 
To define it formally, consider a channel $\cali{C}=\{X,Y,C\}$ as above. The map $C$ induces a map $C': Y\tensor X\rightarrow X$ defined by $C'(y\tensor x)=xC(y)$. Given a state $\omega$ on $X$ the 
dual of the map $C'$ defines a state $\Omega_C$ on $Y\tensor X$: \(\Omega_C(y\tensor x)=\omega(C'(y\tensor x))=C(y|x)\omega(x)\). We call $\Omega_C$ the joint (input-output) state of the channel. A channel is 
useless if $Y$ and $X$ (identified as $Y\tensor \unit$ and $\unit \tensor X$ resp.) are $\Omega_C$-independent. It is easily shown that: {\em a channel $\cali{C}=\{X,Y,C\}$ with input source $(X,\omega)$ is useless iff the matrix $C_{ij}=C(y_j|x_i)$ is of rank 1}.  
The algebraic version of the channel coding theorem assures that it is possible to approximate, in the long run, an arbitrary channel (excepting the useless case) by a lossless one. 

\begin{thm}[Channel coding]\label{thm:ch_coding}
Let $\cali{C}$ be a channel with input algebra $X$ and output algebra $Y$. Let \(\{x_i\}_{i=1}^n\text{ and } \{y_j\}_{j=1}^m\) be atomic bases for $X$ and $Y$ resp. Given a state $\omega$ on $X$, if the channel is not useless then  for each $k$ there are subalgebras \(Y_k\subset \tensor^k Y, X_k\subset \tensor^k X\), a map $C_k: Y_k\rightarrow X_k$ induced by $C$ and  a lossless channel $L_k: Y_k\rightarrow X_k$ such that 
\[\lim_{k\rightarrow \infty} \Omega(|O_{C_k}-O_{L_k}|) = 0 \text{ on } T_k=Y_k\tensor X_k\]
Here $\Omega=\tensor^{\infty}\Omega_C$ and on $\tensor^k Y\tensor \tensor^k Y$ it acts as $\Omega^k=\tensor^k\Omega_C$ where $\Omega_C$ is the state induced by the channel and a given input state $\omega$. Moreover, if $r_k=\text{dim}(X_k)$ then $R=\frac{\log{r_k}}{k}$, called transmission rate, is independent of $k$. 
\end{thm}
Let us clarify the meaning of the above statements. The theorem simply states that on the chosen set of codewords the channel output of $C_k$ induced by the given channel can be made arbitrarily close to that of a lossless channel $L_k$. Since a lossless channel has a definite decision scheme for decoding the choice of $L_k$ is effectively a decision scheme for decoding the original channel's output when the input is restricted to our ``code-book''. This implies probability of error tends to 0 it is possible to choose a set of ``codewords'' which can be transmitted with high reliability. The proof of the theorem \cite{Patra09} uses algebraic arguments only.  The theorem guarantees ``convergence in the mean'' in the appropriate subspace which implies convergence in probability. For a lossless channel the input entropy $H(X)$ is equal to the mutual information. We may think of this as conservation of entropy or information which justifies the term ``lossless''. Since it is always the case that $H(X)-H(X|Y)=I(X,Y)$ the quantity $H(X|Y)$ can be considered the loss due to the channel. The algebraic version of the theorem serves two primary purposes. It gives us the abelian perspective from which we will seek possible extensions to the non-commutative case. Secondly, the channel map $L$ can be used for a decoding scheme. Thus we may think of a coding-decoding scheme for a given channel as a sequence of pairs $(X_k, L_k)$ as above. 

The coding theorems can be extended to more complicated scenarios like ergodic sources and channels with finite memory. We will not pursue these issues further here. But we are confident that these generalizations can be appropriately formulated and proved in the algebraic framework.  
In the preceding sections we have laid the basic algebraic framework for classical information theory. Although, we often confined our discussion to finite-dimensional algebras corresponding to finite sample spaces it is possible to extend it to infinite-dimensional algebras of continuous sample spaces. 
\commentout{
In this regard, a natural question is: can the algebraic formulation replace  Kolmogorov axiomatics based on measure theory? Naively, the answer is no because the assumption of a norm-compete algebra imposes the restriction that the random variables that they represent must be {\em bounded}. Moreover, the GNS construction implies that the algebraic framework is essentially equivalent to (almost) bounded random variables on a locally compact space. In order to deal with the unbounded case we have to go beyond the normed algebra structures. A possible course of action is indicated in the examples given in section \ref{sec:examples}: via the use of a ``cut-off''. A more general approach would be to consider sequences which converge in a topology weaker than the norm topology to elements of a larger algebra. These and other related issues on foundations are deep and merit a separate investigation.}
These topics will be  investigated in the future in the non-commutative setting.  We will delve deeper into these analogies and aim to throw light on some basic issues like quantum Huffman coding \cite{Braunstein}, channel capacities and general no-go theorems among others, once we formulate the appropriate models. 


\begin{thebibliography}{BFGL00}

\bibitem[Ara75]{araki75}
H.~Araki.
\newblock {Relative entropy of states of von Neumann algebras}.
\newblock {\em Publications of the Research Institute for Mathematical
  Sciences}, 11:173--192, 1975.

\bibitem[BFGL00]{Braunstein}
S.~L. Braunstein, C.~A. Fuchs, D.~Gottesman, and H-K. Lo.
\newblock {A quantum analog of Huffman coding}.
\newblock {\em IEEE Trans. Inf. Th.}, 46:1545, 2000.

\bibitem[BKK07]{Beny}
C.~B\^eny, A.~Kempf, and D.~W. Kribs.
\newblock {Quantum error correction of observables}.
\newblock {\em Phys. Rev. A.}, 76:042303, 2007.

\bibitem[CK81]{Ciszar}
I.~Csisz\^ar and K\''{o}rner.
\newblock {\em {Information Theory}}.
\newblock Academic Press, 1981.

\bibitem[CT99]{CoverT}
T.~M. Cover and Joy.~A. Thomas.
\newblock {\em {Elements of Information Theory}}.
\newblock John Wiley {\&} Sons, 1999.

\bibitem[Key02]{Keyl}
M.~Keyl.
\newblock {Fundamentals of quantum information theory}.
\newblock {\em Phys. Rep.}, 369:531--548, 2002.

\bibitem[KR97]{KR1}
R.~V. Kadison and J.~R. Ringrose.
\newblock {\em {Fundamentals of the Theory of Operator Algebras Vol. I}}.
\newblock American Mathematical Society, 1997.

\bibitem[KW06]{Kretschmann}
D.~Kretschmann and R.~F. Werner.
\newblock {Quantum channels with memory}.
\newblock {\em Phys. Rev. A.}, 72:062323, 2006.

\bibitem[PB]{Patra09}
M.~K. Patra and S.~L. Braunstein.
\newblock An algebraic framework for information theory: Classical information.
\newblock \href{http://arxiv.org/abs/0910.1536}{arXiv:0910.1536}.

\bibitem[Rud87]{Rudin}
W.~Rudin.
\newblock {\em {Real and Complex Analysis}}.
\newblock McGraw-Hill, 3 edition, 1987.

\bibitem[Seg54]{Segal54}
I.~E. Segal.
\newblock {Abstract probability spaces and a theorem of Kolmogoroff}.
\newblock {\em Am. J. Math}, 76(3):721--732, 1954.

\bibitem[Seg60]{Segal60}
I.~E. Segal.
\newblock {A note on the concept of entropy}.
\newblock {\em J. Math. Mech.}, 9:623--629, 1960.

\bibitem[Sha48]{Shannon48}
C.~E. Shannon.
\newblock {A mathematical theory of communication}.
\newblock {\em Bell Sys. Tech. Journal}, 27:379--423,623--656, 1948.

\bibitem[Ume62]{Umegaki4}
H.~Umegaki.
\newblock {Conditional expectation in an operator algebra, IV (Entropy and
  Information)}.
\newblock {\em Kodai Mathematical Seminar Reports}, 14:59--85, 1962.

\bibitem[VDN92]{Voic}
D.~Voiculescu, K.~Dykema, and A.~Nica.
\newblock {\em {Free random variables}}.
\newblock {CRM monograph series}. AMS, 1992.

\end{thebibliography}
\end{document}